\documentclass[conference]{IEEEtran}
\usepackage{cite}
\usepackage{amsmath,amssymb,amsfonts}
\usepackage{algorithmic}
\usepackage{graphicx}
\usepackage{textcomp}
\usepackage{amsmath}
\usepackage{esvect}
\usepackage{bm}

\usepackage{amsmath,bbm,amssymb,amsfonts,amstext,amsopn}
\DeclareMathOperator{\Tr}{Tr}
\DeclareMathOperator{\Rank}{Rank}

\usepackage{lipsum,amsmath,multicol}
\usepackage{amsthm}
\usepackage{amssymb}
\usepackage{epstopdf}

\usepackage{amsthm}
\usepackage{cite}
\usepackage{balance}
\allowbreak
\usepackage{url}
\usepackage{epsfig}
\usepackage{setspace}
\usepackage{stmaryrd}
\usepackage{psfrag}	
\usepackage{multirow}
\usepackage{float}
\usepackage{amsthm}

\newcommand{\subparagraph}{}
\usepackage[compact]{titlesec}
\usepackage[english]{babel}
\usepackage[process=auto]{pstool}
\usepackage{amsthm}
\theoremstyle{remark}
\newtheorem{remark}{Remark}
\usepackage{etoolbox}
\usepackage{algorithm}
\usepackage[process=auto]{pstool}
\usepackage{epsfig}
\usepackage{caption}
\usepackage{algorithmic}
\usepackage{xcolor}
\newtheorem{thm}{Theorem}

\begin{document}
	\captionsetup[figure]{labelformat={default},labelsep=period,name={Fig.}}
\title{Resource Allocation for Secure Multi-User Downlink MISO-URLLC Systems} 
\author{ Walid R. Ghanem,
Vahid Jamali, and Robert Schober  \\
Friedrich-Alexander-University Erlangen-Nuremberg, Germany \quad}
\maketitle
\begin{abstract}
In this paper, we study resource allocation algorithm design for secure multi-user downlink ultra-reliable low latency communication (URLLC). To enhance physical layer security (PLS), the base station (BS) is equipped with multiple antennas and artificial noise (AN) is injected by the BS to impair the eavesdroppers' channels. To meet the stringent delay requirements in secure URLLC systems, short packet transmission (SPT) is adopted and taken into consideration for resource allocation design. The resource allocation algorithm design is formulated as an optimization problem for minimization of the total transmit power, while guaranteeing quality-of-service (QoS) constraints regarding the URLLC users' number of transmitted
bits, packet error probability, information leakage, and delay. Due to the non-convexity of the optimization problem, finding a global solution entails a high computational complexity. Thus, we propose a low-complexity algorithm based successive convex approximation (SCA) to find a sub-optimal solution. Our simulation results show that the proposed resource allocation algorithm design ensures the secrecy of the URLLC users' transmissions, and yields significant power savings compared to a baseline scheme. 
\end{abstract}
\section{Introduction} 
 Recently, ultra-reliable low latency communication (URLLC) has received considerable attention from academia and industry. URLLC is required for mission critical applications, such as factory automation, autonomous driving, tactile internet, e-health, and virtual reality \cite{Toward}. URLLC entails strict quality-of-service (QoS) requirements including a low packet error probability (e.g., $10^{-6}$) and a very low latency (e.g., $1\,$ms) \cite{Toward}. In addition, the data packet size is typically small, e.g., around 160~bits \cite{Popovski1}. Unfortunately, existing mobile communication systems cannot meet these requirements. For example, for the long term evolution (LTE) system, the frame duration is 10$\,$ms, which already exceeds the total latency requirement of URLLC applications \cite{Mehdi1}. The main challenges for the design of URLLC systems are the two contradicting requirements of ultra high reliability and low latency.

Furthermore, security is a fundamental issue in wireless communication systems due to the broadcast nature of the wireless medium. Traditionally, communication security was provided by cryptographic encryption methods in the application layer. However, cryptographic methods rely on the assumption that the computational power of eavesdroppers is limited. Considering recent advances in quantum computing, this assumption may no longer be justified. An alternative approach to ensuring secrecy in communication systems is physical layer security (PLS). The principle of PLS is to exploit the physical characteristics of the wireless channel to provide perfect communication secrecy independent of the computational capabilities of the eavesdroppers. In \cite{secureofdma}, the authors studied secure resource allocation and scheduling in orthogonal frequency division multiple access (OFDMA) networks. In \cite{misosec1}, the authors investigated the secrecy rate optimization in a multiple-input single-output (MISO) system in the presence of multiple eavesdroppers. The authors in \cite{misoyan} studied resource allocation for MISO systems, where artificial noise (AN) was injected at the base station (BS) to further enhance PLS. However, the schemes proposed in \cite{secureofdma,misosec1,misoyan} were designed based on the secrecy capacity which assumes infinite length codes \cite{Wyner1}. Hence, such codes are not applicable for URLLC, since URLLC employs short packet transmission (SPT) to achieve low latency.    

Recently, there has been a significant amount of work on the performance limits of SPT. These performance limits provide a relation between the achievable rate, the packet length, and the decoding error probability \cite{strassen,thesis}. The seminal work in \cite{strassen} investigated the limits of SPT for discrete memoryless channels, while the authors in \cite{Polyanskiy} extended this analysis to different types of channels, including the additive white Gaussian noise (AWGN) channel. The authors in \cite{optimal,convexfinite,ghanem1,ghanem2} investigated the resource allocation algorithm design for multi-user downlink URLLC systems. Nevertheless, the existing URLLC designs in \cite{strassen,thesis,Polyanskiy,Quasi,optimal,convexfinite,ghanem1, ghanem2} do not take into account secrecy, and hence, cannot guarantee PLS. The maximum secrecy rate for SPT  over a wiretap channel was provided recently in \cite{wangsecrecy}. Furthermore, in \cite{short1}, the performance limits of secure SPT were studied. However, the author of \cite{short1} focused on the case of a single user system, i.e., a single legitimate user and single eavesdropper. To the best of the authors' knowledge, the resource allocation design for secure multi-user downlink MISO-URLLC systems has not been investigated in the literature, yet.

 Motivated by the above discussion, in this paper, we propose a novel secure resource allocation algorithm design for multi-user downlink MISO-URLLC systems. The resource allocation algorithm design is formulated as an optimization problem with the goal of minimizing the total transmit power subject to the QoS constraints of the URLLC users. The QoS constraints include the minimum number of securely transmitted bits, the maximum packet error probability, the maximum information leakage, and the maximum time for transmission of a packet (i.e., the maximum delay)\footnote{We note that the end-to-end (E2E) delay of data packet transmission comprises  various components including the transmission delay, queueing delay, propagation delay, and routing delay in the backhaul and core networks. In this work, we focus on the transmission delay, which is independent of the other components of the E2E delay.}. The formulated problem is non-convex and finding the global optimum solution requires high computational complexity. Therefore, we develop a low-complexity resource allocation algorithm based on successive convex approximation (SCA) which finds a sub-optimal solution.

\textit{Notations}: In this paper, lower-case letters refer to scalar numbers, while bold lower and upper case letters denote vectors and matrices, respectively. $\Tr{(\mathbf{A})}$ and $\Rank{(\mathbf{A})}$ denote the trace and the rank of matrix $\mathbf{A}$, respectively. $\mathbf{A}\succeq 0$  indicates that matrix $\mathbf{A}$ is positive semi-definite. $\mathbf{A}^{H}$ and ${\mathbf{A}}^{T}$ denote the Hermitian transpose and the transpose of matrix $\mathbf{A}$, respectively. $\mathbb{C}$ is the set of complex numbers. $\mathbf{I}_{N}$ is the $N \times N$ identity matrix. $\mathbb{H}_{N}$  denotes the set of all $N \times N$ Hermitian matrices. $|\cdot|$ and  $\|\cdot\|$ refer to the absolute value of a complex scalar and the Euclidean vector norm, respectively. The distribution of a circularly symmetric complex Gaussian (CSCG) vector with mean vector $\mathbf{x}$ and and covariance matrix $\mathbf{\Sigma}$  is denoted by $\mathcal{CN}(\mathbf{x},\mathbf{\Sigma})$, and $\sim$ stands for ``distributed as". $\mathcal{E}\{\cdot\}$ denotes statistical expectation. $\nabla_{x}f(\mathbf{x})$ denotes the gradient vector of function $f(\mathbf{x})$ and its elements are the partial derivatives of $f(\mathbf{x})$.
\section{System and Channel Models}	
In this section, we present the considered system and channel models.
\subsection{System Model}
We consider a single-cell multi-user downlink system which comprises a BS equipped with $N_{T}$ antennas, $K$ single-antenna URLLC users indexed by $k =\{1,\dots,K\}$, and $J$ single antenna eavesdroppers indexed by $j =\{1,\dots,J\}$, cf. Fig.~\ref{model}. The total bandwidth is $W$. We assume that a resource frame has a duration of $T_{f}$ seconds and is divided into $N$ time slots indexed by $n =\{1,\dots,N\}$. Each time slot $n$ comprises $\bar{n}$ symbol intervals. The value of $\bar{n}$ depends on the system bandwidth $W$ and the total frame duration $T_{f}$, i.e., $\bar{n}=\frac{WT_{f}}{N}$ which is assumed to be integer. Furthermore, perfect channel state information (CSI) is assumed to be available at the BS for resource allocation design\footnote{
	In practice, the BS may not be able to obtain perfect CSI, especially for the eavesdroppes' channels. Hence, the results in this paper can serve as a performance upper bound for system with imperfect CSI.}. We assume that the maximum affordable delay of each user is known at the BS and only users whose delay requirements can potentially be met in the current frame are admitted into the system.
\begin{figure}
	\centering
	\scalebox{1.4}{
		\pstool{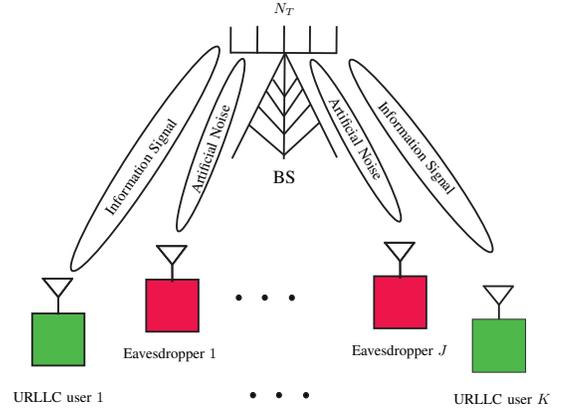}{
		    \psfrag{d}[c][c][0.4]{Meta data}
			\psfrag{a}[c][c][0.4]{$N_{T}$}
		    \psfrag{b}[c][c][0.4]{\text{user 2}}
		    \psfrag{c}[c][c][0.4]{\text{user $K$}}
		    \psfrag{e}[c][c][0.4]{\text{URLLC user $1$}}
		    \psfrag{g}[c][c][0.4]{\text{URLLC user $K$}}
		    \psfrag{m}[c][c][0.5]{BS}
		    \psfrag{f}[c][c][0.4]{Eavesdropper $1$}
		    \psfrag{d}[c][c][0.4]{Eavesdropper $J$}
		    \psfrag{I}[c][c][0.4]{Information Signal}
		    \psfrag{c}[c][c][0.4]{Artificial Noise}
	}}
	\caption{ Multi-user downlink URLLC system with a BS, $K$ URLLC users, and $J$ eavesdroppers.}
	\label{model}
	\vspace{-0.5cm}
\end{figure}

In this paper, we assume linear transmit precoding at the BS. Hence, the signal vector transmitted by the BS in time slot $n$ to the $K$ users is given by
\begin{IEEEeqnarray}{lll}\label{txvector}
	\mathbf{x}[n]=\underbrace{\sum_{k=1}^{K}\mathbf{w}_{k}[n]u_{k}[n]}_{\substack {\text {desired signal}}}+\underbrace{\mathbf{v}[n]}_{\substack {\text {AN}}},
\end{IEEEeqnarray}
where $\mathbf{w}_{k}[n] \in \mathbb{C}^{N_{T} \times 1}$ and $u_{k}[n] \in \mathbb{C}$, $\mathcal{E}\{|u_{k}[n]|^{2}\}=1, \; \forall k=\{1,\dots, K\}$, are the beamforming vector and the transmit symbol of user $k$ in time slot $n$, respectively. Moreover, $\mathbf{v}[n]$ is an AN vector generated by the transmitter to degrade the eavesdroppers' channels. $\mathbf{v}[n]$ is modelled as a complex Gaussian random
vector, $\mathbf{v}[n]\sim \mathcal{CN}(\mathbf{0},\mathbf{V}[n])$,
with covariance matrix $\mathbf{V}[n] \in \mathbb{H}_{N_{T}}$, $\mathbf{V}[n] \succeq 0$. 
\subsection{Channel Model}	\vspace{-0.2cm}
In this paper, we assume a quasi static flat-fading channel whose coherence time exceeds $T_{f}$. Therefore, the channel gains of all users are assumed to be fixed during the entire frame duration. The signals received at user $k$ and eavesdropper $j$ in time slot $n$ are given as follows:
\begin{IEEEeqnarray}{lll}\label{receviedsignal}
	y_{k}[n]=\mathbf{h}^{H}_{k}\mathbf{x}[n]+w_{k}[n], \forall k,
\end{IEEEeqnarray}
\begin{IEEEeqnarray}{lll}\label{eavs}
	y_{j}[n]=\mathbf{g}^{H}_{j}\mathbf{x}[n]+w^{e}_{j}[n], \forall j,
\end{IEEEeqnarray}
where $\mathbf{h}_{k} \in \mathbb{C}^{N_{T} \times 1}$ and $\mathbf{g}_{j} \in \mathbb{C}^{N_{T} \times 1}$ are the channel vectors from the BS to user $k$ and eavesdropper $j$, respectively. Besides, $w_{k}[n]\sim \mathcal{CN}(0,\sigma^{2})$ and $w^{e}_{j}[n]\sim \mathcal{CN}(0,\sigma^{2})$ are the complex AWGN\footnote{ Without loss of generality, we assume that the noise processes at all receivers have identical variances.} at user $k$ and eavesdropper $j$, respectively.
 By substituting (\ref{txvector}) into (\ref{receviedsignal}), the signal received at user $k$ in time slot $n$ is given as follows: 
\begin{IEEEeqnarray}{lll}
	y_{k}[n]=\mathbf{h}_{k}^{H}\left(\sum_{l=1}^{K}\mathbf{w}_{l}[n]u_{l}[n]+\mathbf{v}[n]\right)+w_{k}[n],\\ 
	=\underbrace{\mathbf{h}_{k}^{H}\mathbf{w}_{k}[n]u_{k}[n]}_{\substack {\text {desired signal}}}+\underbrace{\sum_{l\neq k }^{K}\mathbf{h}^{H}_{k}\mathbf{w}_{l}[n]u_{l}[n]}_{\substack {\text {multi-user interference (MUI)}}}+\underbrace{\mathbf{h}_{k}^{H}\mathbf{v}[n]}_{\substack {\text {AN}}}+w_{k}[n].\nonumber
\end{IEEEeqnarray}
 Similarly, the received signal for user $k$ at eavesdropper $j$ in time slot $n$ can be expressed as follows:
\begin{IEEEeqnarray}{lll}
	y_{j,k}[n]=\mathbf{g}_{j}^{H}\left(\sum_{l=1}^{K}\mathbf{w}_{l}[n]u_{l}[n]+\mathbf{v}[n]\right)+w^{e}_{j}[n],\\ 
	=\underbrace{\mathbf{g}_{j}^{H}\mathbf{w}_{k}[n]u_{k}[n]}_{\substack {\text {user's $k$ signal}}}+\underbrace{\sum_{l\neq k }^{K}\mathbf{g}^{H}_{j}\mathbf{w}_{l}[n]u_{l}[n]}_{\substack {\text { MUI}}}+\underbrace{\mathbf{g}_{j}^{H}\mathbf{v}[n]}_{\substack {\text {AN}}}+w^{e}_{j}[n].\nonumber
\end{IEEEeqnarray}
Moreover, for future use, we define the signal-to-interference-plus-noise-ratio (SINR) for user $k$ and eavesdropper $j$ in time slot $n$, respectively, as follows: 
\begin{IEEEeqnarray}{lll}{\label{rho5}}
	\gamma_{k}[n]=\frac{|\mathbf{h}_{k}^{H}\mathbf{w}_{k}[n]|^{2}}{\sum^{K}_{l\neq k}|\mathbf{h}_{k}^{H}\mathbf{w}_{l}[n]|^{2}+\Tr(\mathbf{h}_{k}\mathbf{h}_{k}^{H}\mathbf{V}[n])+\sigma^{2}},
\end{IEEEeqnarray}
\begin{IEEEeqnarray}{lll}{\label{rho6}}
	\gamma_{j,k}[n]=\frac{|\mathbf{g}_{j}^{H}\mathbf{w}_{k}[n]|^{2}}{\sum^{K}_{l\neq k}|\mathbf{g}_{j}^{H}\mathbf{w}_{l}[n]|^{2}+\Tr(\mathbf{g}_{j}\mathbf{g}_{j}^{H}\mathbf{V}[n])+\sigma^{2}}.
\end{IEEEeqnarray}
In this paper, we treat the MUI caused by the signals of other users as noise. Moreover, we assume a normalized channel model such that $\sigma^{2}=1$ holds. 
\section{Resource Allocation Problem Formulation}\vspace{-0.2cm}
In this section, we introduce the achievable secrecy rate for SPT and the QoS requirements of the URLLC users. Furthermore, we formulate the proposed resource allocation optimization problem for secure multi-user downlink MISO-URLLC systems.
\subsection{Achievable Secrecy Rate for SPT}\vspace{-0.2cm}
In their seminal works, Wyner \cite{Wyner1} and Csisz\'ar and K\"orner \cite{Csiszar1} characterized the secrecy capacity. It was shown that both the error probability and the information leakage can be made arbitrarily small as long
as the transmission rate is below the secrecy capacity and the data are mapped to
sufficiently long codewords, i.e., the packet length goes to infinity. Because of the latter condition, the secrecy capacity cannot be used for resource allocation design in secure URLLC systems, as URLLC systems have to employ short packets to achieve low latency. Furthermore, as a result of using short packets, decoding errors and information leakage become unavoidable. The achievable secrecy rate for SPT was analyzed in \cite{wangsecrecy} and a closed-form approximation, the so-called normal approximation, was developed for the AWGN channel. Mathematically, the maximum number of secret communication bits transmitted in a packet of $L$ symbols to a legitimate receiver with error probability $\epsilon$ and information leakage $\delta$ can be approximated as follows\cite[Eq. (108)]{wangsecrecy}:
\begin{IEEEeqnarray}{lll}\label{normalapproximation}
\hspace{-0.5cm}	\bar{B}=L\bar{C}_{s}-aQ^{-1}(\epsilon)\sqrt{LV_{m}}-aQ^{-1}(\delta)\sqrt{LV_{e}},
\end{IEEEeqnarray}
where $a=\log(\text{e})$, and the secrecy capacity is given by
\begin{IEEEeqnarray}{lll}\label{cs}
\bar{C}_{s}=\log_{2}(1+\gamma_{m})-\log_{2}(1+\gamma_{e}).
\end{IEEEeqnarray}
Here, $Q^{-1}(\cdot)$ is the inverse of the Gaussian Q-function,  $Q(x)=\frac{1}{\sqrt{2\pi}}\int_{x}^{\infty}\text{exp}{\left(-\frac{t^{2}}{2}\right)}\text{d}t$, and
$V_{m}$ and $V_{e}$ are the channel dispersions of the legitimate and eavesdropper channels, respectively. For the complex AWGN channel, the channel dispersion is given as follows \cite{thesis}:
\begin{IEEEeqnarray}{lll}\label{dispersion}
	V_{i}=\bigg(1-\frac{1}{(1+\gamma_{i})^2}\bigg), \forall i \in \{m,e\},
\end{IEEEeqnarray} 
where $\gamma_{i}$ is the received SINR at receiver $i$, i.e., at the legitimate receiver or the eavesdropper.

 For (\ref{normalapproximation}), it was assumed that all symbols in the packet have the same SINR. However, for the case where the SINRs of different symbols may not be identical, using a similar approach as in \cite[Eq. (4.277)]{thesis},\cite[Fig. 1]{Erseghe1}, the maximum number of secret communication bits for SPT can be expressed as follows:
\begin{IEEEeqnarray}{lll}\label{secrate}
	\hspace{-0.6cm}	B^{\mathrm{sec}}=LC_{s}-aQ^{-1}(\epsilon)\sqrt{\sum_{l=1}^{L}V_{m}[l]}-aQ^{-1}(\delta)\sqrt{\sum_{l=1}^{L}V_{e}[l]},
\end{IEEEeqnarray}
where the secrecy capacity in this case is given by:
\begin{IEEEeqnarray}{lll}\label{cs}
	C_{s}=\frac{1}{L}\bigg(\sum_{l=1}^{L}\log_{2}(1+\gamma_{m}[l])-\sum_{l=1}^{L}\log_{2}(1+\gamma_{e}[l])\bigg),
\end{IEEEeqnarray} 

Here, $V_{m}[l]$, $V_{e}[l]$, $\gamma_{m}[l]$, and $\gamma_{e}[l]$ denote the channel dispersion for the legitimate receiver, the channel dispersion for the eavesdropper, the SINR at the legitimate receiver, and the SINR at the eavesdropper in symbol interval $l$, respectively.       

In this paper, the resource allocation algorithm design for secure MISO-URLLC is based on (\ref{secrate}).  
\subsection{QoS Requirements of URLLC Users}\vspace{-0.1cm}
The QoS requirements of URLLC user $k$ comprise a minimum number of securely received bits, denoted by $B_{k}^{\text{req}}$, a maximum packet error probability denoted by $\epsilon_{k}$, the maximum number of time slots available for transmission of the user's packet, denoted by $D_{k}$, and the maximum information leakage $\delta_{j,k}$ to eavesdropper $j$. According to (\ref{secrate}), the total number of secret communication bits transmitted over the resources allocated to user $k$ in the $N$ available time slots in the presence of $J$ eavesdroppers can be written as:
\begin{multline}\label{BT}
B^{\mathrm{sec}}_{k}(\mathbf{w}_{k},\mathbf{V})=\bar{n} \sum_{n=1}^{N}\log(1+\gamma_{k}[n])- aQ^{-1}(\epsilon_{k})\bigg(\bar{n}{\sum_{n=1}^{N} V_{k}[n]}\bigg)^{\frac{1}{2}}\\-\hspace{-0.4cm}\max_{j\in\{1,2,\cdots, J\}}\bigg[\bar{n}\sum_{n=1}^{N}\log(1+\gamma_{i,k}[n])
\\+aQ^{-1}(\delta_{j,k})\bigg(\bar{n}{\sum_{n=1}^{N} V_{j,k}[n]}\bigg)^{\frac{1}{2}}\bigg],\hspace{-0.2cm}
\end{multline}
where $\mathbf {V}$ denotes the collection of optimization variables $\mathbf{V}[n], \forall n$, $\mathbf{w}_{k}$ denotes the collection of optimization variables $\mathbf{w}_{k}[n],$ $\forall n$, and
\begin{IEEEeqnarray}{lll}\label{dispersion}
	V_{k}[n]=1-\frac{1}{(1+\gamma_{k}[n])^2},\nonumber	V_{j,k}[n]=1-\frac{1}{(1+\gamma_{j,k}[n])^2}.\nonumber
\end{IEEEeqnarray} 

 Moreover, to meet the delay requirement of user $k$, we assign all symbols of user $k$ to the first $D_{k}$ time slots. This means that the value of $\gamma_{k}[n]$, and thus, the value of $\mathbf{w}_{k}[n]$, should be zero for user $k$ if $n>D_{k}$. In other words, users requiring low latency are assigned resources at the beginning of the frame by controlling the value of $\mathbf{w}_{k}[n], \forall k,n$. Then, a user can start decoding as soon as it has received all symbols that contain its data, i.e., after $D_{k}$ time slots.  
\subsection{Optimization Problem Formulation}\vspace{-0.1cm}
In the following, we formulate the proposed resource allocation optimization problem for the minimization of the total transmit power under constraints on the QoS of each user regarding the received number of secrecy bits, the reliability, the latency, and the information leakage. In particular, the beamforming and the AN policies are determined by solving the following optimization
problem:
\begin{IEEEeqnarray}{lll}\label{op1}& \underset { \mathbf {w}, \mathbf {V} \in \mathbb{H}_{N_{T}}}{ \mathop {\mathrm {minimize}}\nolimits }~ \sum_{k=1}^{K}\sum_{n=1}^{N}ٌٌ\|\mathbf{w}_{k}[n]\|^{2}+\sum_{n=1}^{N}\Tr(\mathbf{V}[n]) \\& \mathrm {s.t.}~\mathrm {C1:}\; \nonumber B^{\mathrm{sec}}_{k}(\mathbf{w}_{k},\mathbf{V}) \geq B^{\text{req}}_{k}, \forall k, \nonumber		
	\\& \;\; \quad \mathrm{C2:}\; \mathbf{w}_{k}[n]=0, \forall n > D_{k}, \forall k,\nonumber
	\\& \;\; \quad \mathrm {C3:}\; \mathbf {V}[n] \succeq 0 , \forall n,\nonumber
 \end{IEEEeqnarray}
where $\mathbf{w}$ denotes the collection of optimization variables $\mathbf{w}_{k}, \forall k$. The optimization problem in (\ref{op1}) is non-convex. The non-convexity arises from constraint $\mathrm{C1}$ which involves the non-convex SINR expressions and the non-convex secrecy rate formula for SPT. In general, there is no systematic method for solving non-convex optimization problems in polynomial time. Nevertheless, in the next section, we propose a low-complexity algorithm to find a sub-optimal solution for optimization problem (\ref{op1}). \vspace{-0.3cm}
\section{Solution of the Problem}\vspace{-0.2cm}
In the following, we propose a low-complexity resource allocation algorithm based on SCA to solve optimization problem (\ref{op1}) and to obtain a sub-optimal resource allocation policy with low computational complexity. The proposed algorithm design tackles problem (\ref{op1}) in three main
steps as outlined in the following. First, we transform the problem into a more tractable form using semi-definite programming (SDP). Second, we apply a series of transformations based on auxiliary variables and Taylor series expansion to obtain an approximated convex problem. Finally, a penalized SCA algorithm is proposed to iteratively solve the approximated convex problem, and to find a sub-optimal  solution for problem (\ref{op1}).    \vspace{-0.2cm}
\subsection{Semi-Definite Programming}\vspace{-0.2cm}
To facilitate solving problem (\ref{op1}) using SDP, we define $\mathbf{W}_{k}[n]=\mathbf{w}_{k}[n]\mathbf{w}_{k}^{H}[n]$, $\mathbf{H}_{k}=\mathbf{h}_{k}\mathbf{h}_{k}^{H}$, and $\mathbf{G}_{j}=\mathbf{g}_{j}\mathbf{g}_{j}^{H}$. Therefore, problem (\ref{op1}) can be rewritten in equivalent form as follows:  
\begin{IEEEeqnarray}{lll}\label{op1a}& \underset { \mathbf {W},\mathbf {V} \in \mathbb{H}_{N_{T}}}{ \mathop {\mathrm {minimize}}\nolimits }~ G(\mathbf {W},\mathbf {V}) \\& \;\;\mathrm {s.t.}~\mathrm {C1:} B^{\mathrm{sec}}_{k}(\mathbf{W}_{k},\mathbf{V}) \geq B^{\text{req}}_{k}, \forall k, \nonumber		
	\\&  \qquad\mathrm {C2:} \Tr{(\mathbf{W}_{k}[n])}=0, \forall n > D_{k},  \forall k\nonumber, 
		\\& \qquad \mathrm {C3:} \mathbf {V}[n] \succeq 0 , \forall n,\nonumber
	\\&\qquad \mathrm {C4:} {\mathbf{W}_{k}}[n]\succeq 0,\forall k,n,\nonumber
	\\&\qquad \mathrm {C5:} \Rank(\mathbf{W}_{k}[n]) \leq 1,\forall k,n,\nonumber
\end{IEEEeqnarray}
where $\mathbf{W}_{k}$ is the collection of optimization variables $\mathbf{W}_{k}[n],\forall n$, and $\mathbf{W}$ is the collection of optimization variables $\mathbf{W}_{k}, \forall k$. ${\mathbf{W}_{k}}[n]\succeq 0$ and $\Rank({\mathbf{W}_{k}[n]})\leq 1$, $\forall k,n,$ in constraints $\mathrm {C4}$ and $\mathrm {C5}$ are imposed to ensure that $\mathbf{W}_{k}[n]=\mathbf{w}_{k}[n]\mathbf{w}_{k}^{H}[n]$ holds after optimization, and 	\vspace{-0.2cm}  
$$ G(\mathbf {W},\mathbf {V})=\sum_{k=1}^{K}\sum_{n=1}^{N}\Tr{(\mathbf{W}_{k}[n])}+\sum_{n=1}^{N}\Tr(\mathbf{V}[n]).$$
Moreover, to facilitate the solution of problem (\ref{op1a}), constraint $\mathrm{C1}$ is rewritten as follows:
  \begin{multline}{\label{c1}}
 \mathrm{C1}: {R}_{k}(\mathbf{W}_{k}, \mathbf{V})-{V}_{k}(\mathbf{W}_{k}, \mathbf{V})-\\ \max_{j\in\{1,2,\cdots, J\}}{C}^{E}_{j,k}(\mathbf{W}_{k}, \mathbf{V})\geq B_{k}^{\text{req}}, \forall k,
   \end{multline}	\vspace{-0.5cm}
   where
\begin{equation}{\label{f1}}\hspace{-1cm}
{R}_{k}(\mathbf{W}_{k}, \mathbf{V})=\bar{n} \sum_{n=1}^{N}\log(1+\gamma_{k}[n]),
   \end{equation}	\vspace{-0.5cm}
      \begin{equation}{\label{v1}}\hspace{-0.6cm}
   {V}_{k}(\mathbf{W}_{k}, \mathbf{V})= Q^{-1}(\epsilon_{k})\bigg(\bar{n}{\sum_{n=1}^N V_{k}[n]}\bigg)^{\frac{1}{2}},
   \end{equation}
   and   	\vspace{-0.2cm}
\begin{multline}{\label{f2}}
 {C}^{E}_{j,k}(\mathbf{W}_{k}, \mathbf{V})=\bar{n}\sum_{n=1}^{N}\log(1+\gamma_{j,k}[n])\\+Q^{-1}(\delta_{j,k})\bigg(\bar{n}{\sum_{n=1}^N V_{j,k}[n]}\bigg)^{\frac{1}{2}},
   \end{multline}
 where	\vspace{-0.4cm}
 \begin{equation}{\label{rho5a}}
 	\gamma_{k}[n]=\frac{\Tr{(\mathbf{H}_{k}\mathbf{W}_{k}[n])}}{\sum^{K}_{k\neq l}\Tr{(\mathbf{H}_{k}\mathbf{W}_{l}[n])}+\Tr(\mathbf{H}_{k}\mathbf{V}[n])+1},
 \end{equation} 
 \begin{equation}{\label{rho6a}}
 	\gamma_{j,k}[n]=\frac{\Tr{(\mathbf{G}_{j}\mathbf{W}_{k}[n])}}{\sum^{K}_{k\neq l}\Tr{(\mathbf{G}_{j}\mathbf{W}_{l}[n])}+\Tr(\mathbf{G}_{j}\mathbf{V}[n])+1}.
 \end{equation}
\subsection{Problem Transformation}
  In the following, we tackle the non-convexity of problem (\ref{op1a}) arising from non-convex constraints $\mathrm{C1}$ and $\mathrm{C5}$. For constraint $\mathrm{C1}$, we first apply a series of transformations based on auxiliary slack variables. Subsequently, Taylor series expansion is used to find a convex approximation for the non-convex parts. For constraint $\mathrm{C5}$, we resort to the well-known semi-definite relaxation (SDR). Using slack variables $\tau_{k}, \forall k$, we rewrite constraint $\mathrm{C1}$ equivalently as follows:
\begin{IEEEeqnarray}{lll}\label{c1d}
	\mathrm{C1a}: R_{k}(\mathbf{W}_{k}, \mathbf{V})-V_{k}(\mathbf{W}_{k}, \mathbf{V})-\tau_{k}\geq B_{k}^{\text{req}}, \forall k,
\end{IEEEeqnarray}	\vspace{-0.7cm}
\begin{IEEEeqnarray}{lll}\label{c1f}\hspace{-2.5cm}
	\mathrm{C1b}:\tau_{k}\geq  {C}^{E}_{j,k}(\mathbf{W}_{k}, \mathbf{V}), \forall j, \forall k.
\end{IEEEeqnarray}
One reason for the non-convexity of constraints (\ref{c1d}) and (\ref{c1f}) is the structure of the SINRs in (\ref{rho5a}) and (\ref{rho6a}). To tackle this non-convexity, we introduce  auxiliary variables $a_{k}[n], \forall k,n,$ and $b_{j,k}[n], \forall j,k,n,$ to bound the SINRs in (\ref{rho5a}) and (\ref{rho6a}), respectively. By substituting $a_{k}[n], \forall k,n,$ and $b_{j,k}[n], \forall j,k,n,$ for $\gamma_{k}[n], \forall k,n, $ and $\gamma_{j,k}[n], \forall j,k,n$, respectively, in functions $R_{k}(\mathbf{W}_{k}, \mathbf{V})$, $V_{k}(\mathbf{W}_{k}, \mathbf{V})$, and $C^{E}_{j,k}(\mathbf{W}_{k}, \mathbf{V})$, we get new functions $R_{k}(\mathbf{a}_{k})$, $V_{k}(\mathbf{a}_{k})$, and $C^{E}_{j,k}(\mathbf{b}_{j,k})$, respectively, where $\mathbf{a}_{k}$ and $\mathbf{b}_{j,k}$ are the collections of optimization variables $a_{k}[n], \forall n$, and $b_{j,k}[n], \forall n$, respectively. This leads to the equivalent optimization problem: \vspace{-0.2cm} 
\begin{IEEEeqnarray}{lll}\label{op1b}& \underset {\mathbf {W}, \mathbf {V} \in \mathbb{H}_{N_{T}}, \boldsymbol{\tau}, \boldsymbol{a}, \boldsymbol{b}}{ \mathop {\mathrm {minimize}}\nolimits }~ G(\mathbf {W},\mathbf {V}) \\& \;\mathrm {s.t.}~\mathrm {C1a:} R_{k}(\mathbf{a}_{k})-V_{k}(\mathbf{a}_{k})-\tau_{k}\geq B_{k}^{\text{req}}, \forall k,\nonumber \\& \qquad \mathrm{C1b:} \tau_{k}\geq  {C}^{E}_{j,k}(\mathbf{b}_{j,k}), \forall j, k, \nonumber		
	\\&\qquad\mathrm {C2:} \Tr{(\mathbf{W}_{k}[n])}=0, \forall n \geq D_{k},  \forall k\nonumber,
			\\& \qquad \mathrm {C3:} \mathbf {V}[n] \succeq 0 , \forall n,\nonumber 
	\\&\qquad \mathrm {C4:} {\mathbf{W}_{k}}[n]\succeq 0,\forall k,n,\nonumber
	\\&\qquad \mathrm {C5:} \Rank(\mathbf{W}_{k}[n]) \leq 1,\forall k,n,\nonumber
	\\&\qquad \mathrm{C6:} a_{k}[n] \leq 	\gamma_{k}[n],\forall k,n,\nonumber
	\\&\qquad \mathrm{C7:} b_{j,k}[n] \geq \gamma_{j,k}[n], \forall j,k,n,\nonumber
\end{IEEEeqnarray}	
where $\boldsymbol{\tau}$, $\boldsymbol{a}$, and $\boldsymbol{b}$ denote the collections of optimization variables $\tau_{k}, \forall k$, $a_{k}[n], \forall k,n,$ and $b_{j,k}[n], \forall j,k,n,$ respectively.  
\begin{thm}
Optimization problems (\ref{op1a}) and (\ref{op1b}) are equivalent and share the same solution for $\mathbf {W}$ and $\mathbf {V}$. 
\end{thm}
\begin{proof}
Optimization problem (\ref{op1b}) can be formulated as a monotonic optimization problem. As a result, constraints $\mathrm{C6}$ and $\mathrm{C7}$ in (\ref{op1b}) have to hold with equality. Hence, it can be shown that (\ref{op1a}) and (\ref{op1b}) are equivalent. A more detailed proof is omitted here due to the space limitation. See \cite{ghanem2} for a similar proof.   
\end{proof}
 To facilitate the application of SCA\cite{lin1satellite,ghanem2,yan}, we use Taylor series to approximate for the non-convex terms in constraints $\mathrm{C1a}$ and $\mathrm{C1b}$ in (\ref{op1b}). This leads to the following convex constraints:
\begin{IEEEeqnarray}{lll}\label{c1da}\hspace{-1cm}
	\overline{\mathrm{C1a}}: R_{k}(\mathbf{a}_{k})-\tilde{V}_{k}(\mathbf{a}_{k})-\tau_{k}\geq B_{k}^{\text{req}}, \forall k,
\end{IEEEeqnarray}
\begin{IEEEeqnarray}{lll}\label{c1fb}\hspace{-2.8cm}
	\overline{\mathrm{C1b}}:\tau_{k}\geq \tilde{C}^{E}_{j,k}(\mathbf{b}_{j,k}), \forall j,k,
\end{IEEEeqnarray}
where
\begin{IEEEeqnarray*}{lll}\label{c1da}\hspace{-1cm}
	\tilde{V}_{k}(\mathbf{a}_{k})={V}_{k}(\mathbf{a}^{(i)}_{k})+\nabla_{\mathbf{a}_{k}}({V}_{k})(\mathbf{a}_{k}-\mathbf{a}_{k}^{(i)}), \forall k,
\end{IEEEeqnarray*}
\begin{IEEEeqnarray*}{lll}\label{c1db}\hspace{-0.6cm}
	\tilde{C}^{E}_{j,k}(\mathbf{b}_{j,k})={C}^{E}_{j,k}(\mathbf{b}_{j,k}^{(i)})+\nabla_{\mathbf{b}_{j,k}}({C}^{E}_{j,k})(\mathbf{b}_{j,k}-\mathbf{b}_{j,k}^{(i)}), \forall j,k,
\end{IEEEeqnarray*}
and $\mathbf{a}^{(i)}_{k}$ and $\mathbf{b}^{(i)}_{j,k}$ are the feasible points from the previous SCA iteration. 

Next, we tackle the non-convex constraints $\mathrm{C6}$ and $\mathrm{C7}$. For $\mathrm{C6}$, we define slack optimization variables $q_{k}[n],\forall k, n$, and $z_{k}[n], \forall k, n,$ to lower bound the numerator and to upper bound the denominator of $\gamma_{k}[n], \forall k,n,$ in constraint $\mathrm{C6}$, respectively, as follows:  
\begin{equation}\label{qq1}\hspace{-4cm}
\mathrm{C6a}: \Tr(\mathbf{H}_{k}^{H}\mathbf{W}_{k}[n])\geq {q}_{k}^{2}[n], \forall k,n, 
\end{equation}
\begin{multline}\label{qq2}\hspace{-0.4cm} 
\mathrm{C6b}:\sum_{k\neq  l}\Tr(\mathbf{H}_{k}^{H}\mathbf{W}_{l}[n])\\+\Tr(\mathbf{H}^{H}_{k}\mathbf{V}[n])+1\leq z_{k}[n], \forall k,n,
\end{multline}
\vspace{-0.4cm}
\begin{IEEEeqnarray}{lll}\label{qq3}\hspace{-5.7cm}
\mathrm{C6c}: \frac{{q}_{k}^{2}[n]}{z_{k}[n]}\geq a_{k}[n],\forall k,n.
\end{IEEEeqnarray}
The new constraints $\mathrm{C6a}$ and $\mathrm{C6b}$ are convex, however, constraint $\mathrm{C6c}$ in (\ref{qq3}) is still non-convex. Thus, we use a Taylor series to get a first order approximation as follows:  
\begin{multline}\label{Tay1}\hspace{-0.4cm}
\overline{\mathrm{C6c}}:2\left(q^{(i)}_{k}[n]/z^{(i)}_{k}[n]\right)q_{k}[n]\\-{\left(q^{(i)}_{k}[n]/z^{(i)}_{k}[n]\right)}^{2}z_{k}[n]\geq a_{k}[n],\forall k,n,
\end{multline}
where $q^{(i)}_{k}[n]$ and $z^{(i)}_{k}[n]$ are feasible points from the previous SCA iteration. 

Similarly, for non-convex constraint $\mathrm{C7}$, we introduce auxiliary variables $f_{j,k}[n], \forall j,k,n,$ to obtain the following equivalent constraints:
\begin{multline}\label{s1}\hspace{-0.4cm}
\mathrm{C7a:} b_{j,k}[n]\left(\sum_{k\neq l}\Tr(\mathbf{G}_{e}^{H}\mathbf{W}_{l}[n])+\Tr(\mathbf{G}^{H}_{e}\mathbf{V}[n])+1\right)\\\geq \big(f_{j,k}[n]\big)^{2},\forall j,k,n,
\end{multline}
\begin{IEEEeqnarray}{lll}\label{s2}\hspace{-3.5cm}
\mathrm{C7b:}	\big(f_{j,k}[n]\big)^{2} \geq
	\Tr(\mathbf{G}_{e}^{H}\mathbf{W}_{k}[n]),\forall j,k,n,
\end{IEEEeqnarray}
 and by using the S-Procedure\cite{Boyed}, we can transform constraint $\mathrm{C7a}$ in (\ref{s1}) into the following positive semi-definite constraint:  
\begin{multline}\label{gama1}
	 \mathrm{C7a:} \begin{bmatrix} b_{j,k}[n] & f_{j,k}[n] \\ f_{j,k}[n] & \sum_{k\neq l}\Tr(\mathbf{G}_{j}^{H}\mathbf{W}_{l}[n])+\Tr(\mathbf{G}^{H}_{j}\mathbf{V}[n])+1 \end{bmatrix}\\
	 \geq  \mathbf{0},\forall j,k,n.
\end{multline}
Now, we use again a Taylor series to obtain a first order approximation of the right hand side of (\ref{s2}) as follows:
 \begin{multline}
	\overline{\mathrm{{C7b}}}:\big(f_{j,k}^{(i)}[n]\big)^{2}+2f_{j,k}^{(i)}[n](f_{j,k}[n]-f^{(i)}_{j,k}[n])\\ \geq
	\Tr(\mathbf{G}_{j}^{H}\mathbf{W}_{k}[n]),\forall j,k,n,
\end{multline}
where $f_{j,k}^{(i)}[n], \forall j,k,n,$ are feasible points from the previous SCA iteration. 

The only remaining non-convex constraint in (\ref{op1b}) is the rank constraint. By applying SDR, i.e., by dropping the rank constraint, we obtain the following relaxed optimization problem:
\begin{IEEEeqnarray}{lll}\label{op1c}& \hspace{-0.5cm} \underset { \mathbf {W},\mathbf {V} \in \mathbb{H}_{N_{T}}, \boldsymbol{\tau}, \boldsymbol{a}, \boldsymbol{b}, \boldsymbol{q}, \boldsymbol{z}, \boldsymbol{f}}{ \mathop {\mathrm {minimize}}\nolimits }~ G(\mathbf {W},\mathbf {V}) \\& \;\;{s.t.}~ \overline{\mathrm{C1a}}, \overline{\mathrm{C1b}}, \mathrm{C2}, \mathrm{C3}, \mathrm{C4}, \mathrm{C6a}, \mathrm{{C6b}}, \overline{\mathrm{C6c}}, \mathrm{C7a}, \overline{\mathrm{C7b}} \nonumber,   
\end{IEEEeqnarray}
where $\boldsymbol{q}$, $\boldsymbol{z}$, and $\boldsymbol{f}$ denote the collection of optimization variables $q_{k}[n], \forall k,n,$ $z_{k}[n], \forall k,n,$ and $f_{j,k}[n], \forall j,k,n,$ respectively. The convex optimization problem in (\ref{op1c}) can be efficiently solved by standard convex solvers such as CVX \cite{cvx}. A solution of problem (\ref{op1b}) can be found by solving (\ref{op1c}) in an iterative manner, where the solution of (\ref{op1c}) in iteration ${i}$ is used as the initial point for the next iteration ${i}+1$. This leads to a sequence of improved feasible solutions until convergence to a sub-optimal solution (stationary point) of problem (\ref{op1c}), or equivalently problem (\ref{op1a}), in polynomial time \cite{yan, ghanem1, ghanem2,localcon}. However, in general, it is difficult to find initial points satisfying the constraints in (\ref{op1c}). Therefore, we address this issue by penalizing optimization problem (\ref{op1c}) in the following subsection. 
\subsection{Proposed Penalized Algorithm}
In order to solve (\ref{op1c}) using SCA, we require feasible initial points that satisfies the constraints in (\ref{op1c}), especially constraint  $\mathrm{C1a}$. Since it is not easy to find such feasible initial points, we propose an algorithm which is based on penalizing optimization problem (\ref{op1c}) when the constraints are violated. The basic idea is to relax the considered problem by adding slack variables $\theta_{k} \geq 0, \forall k,$
 to  constraint $\mathrm{C1a}$ and penalizing the sum of the violations of the constraints. Thereby, using this technique, optimization problem (\ref{op1c}) can be rewritten as follows:
\begin{IEEEeqnarray}{lll}\label{op1cb}& \hspace{-0.5cm} \underset { \mathbf {W},\mathbf {V} \in \mathbb{H}_{N_{T}}, \boldsymbol{\tau}, \boldsymbol{a}, \boldsymbol{b}, \boldsymbol{q}, \boldsymbol{z}, \boldsymbol{f}, \boldsymbol{\theta}}{ \mathop {\mathrm {minimize}}\nolimits }~ G(\mathbf {W},\mathbf {V})+\beta^{(i)}\sum_{k=1}^{K}\theta_{k} \\& \mathrm {s.t.}~ \overline{\mathrm{C1a}}: R_{k}(\mathbf{a}_{k})-\tilde{V}_{k}(\mathbf{a}_{k})-\tau_{k}+\theta_{k}\geq B_{k}^{\text{req}}, \forall k, \nonumber\\& \overline{\mathrm{C1b}}, \mathrm{{C2}}, \mathrm{{C3}}, \mathrm{{C4}}, \mathrm{{C6a}}, \mathrm{{C6b}}, \overline{\mathrm{C6c}}, \mathrm{{C7a}}, \overline{\mathrm{C7b}} \nonumber,\\& \mathrm{C8}:\theta_{k} \geq 0, \forall k,  \nonumber 
\end{IEEEeqnarray}
where $\beta^{(i)}$ is the penalizing weight in iteration $i$,  and $\boldsymbol{\theta}$ is the collection of slack variables $\theta_{k}, \forall k$.  An iterative algorithm for solving (\ref{op1}) by repeatedly solving (\ref{op1cb}) is provided in \textbf{Algorithm} \ref{sca}. In the first iteration, by choosing a small penalty weight $\beta^{(1)}>0$, we allow the QoS constraint to be violated such that the feasible set is large. Then, in each subsequent iteration ${i}$, we use the solution from the previous iteration as initial point, increase the penalty weight $\beta^{(i)}$, and solve problem (\ref{op1cb}) again. Continuing this iterative procedure eventually yields solutions where $\theta_{k}=0, \forall k,$ holds, i.e., (\ref{op1c}) becomes equivalent to (\ref{op1cb}). Moreover, it was shown in \cite{pccp} that for sufficiently large values of $\beta_{\text{max}}$, \textbf{Algorithm} \ref{sca} will yield the optimal solution for problem (\ref{op1c}). The maximum value $\beta_{\text{max}}$ for the penalty weight is imposed to avoid numerical instability. 

\begin{algorithm}[t]
	\caption{Penalized Successive Convex Approximation}
	1: {Initialize:} random initial points $\boldsymbol{a}^{(1)}$, $\boldsymbol{b}^{(1)}$, $\boldsymbol{q}^{(1)}$, $\boldsymbol{z}^{(1)}$, $\boldsymbol{f}^{(1)}$, set iteration index $i=1$, and maximum number of iterations $I_{\text{max}}$, initial penalty factor $\beta^{(1)}\gg {1}$, $\beta_{\text{max}}$, $\eta >1$.\\
	2: \textbf{Repeat}\\
	3: Solve convex problem (\ref{op1cb}) for given   $\boldsymbol{a}^{(i)}$, $\boldsymbol{b}^{(i)}$, $\boldsymbol{q}^{(i)}$,  $\boldsymbol{z}^{(i)}$, and $\boldsymbol{f}^{(i)}$, and store the intermediate solutions  $\boldsymbol{a}, \boldsymbol{b}, \boldsymbol{q}, \boldsymbol{z},$ and $\boldsymbol{f}$\\
	4: Set ${i}={i}+1$ and update  $\boldsymbol{a}^{({i})}=\boldsymbol{a}$, $\boldsymbol{b}^{({i})}=\boldsymbol{b}$, $\boldsymbol{q}^{({i})}=\boldsymbol{q}$, $\boldsymbol{z}^{({i})}=\boldsymbol{z}$,
	$\boldsymbol{f}^{({i})}=\boldsymbol{f}$, and 
	 $\beta^{(i)}=\min(\eta\beta^{(i-1)},\beta_{\text{max}})$. \\
	6: \textbf{Until} convergence or $i=I_{\text{max}}$.\\
	7: {Output:} $\mathbf{W}^{*}=\mathbf{W}$ and $\mathbf{V}^{*}=\mathbf{V}$.
	\label{sca}
\end{algorithm}
\begin{remark}
We note that the optimal solution of $\mathbf{W}_{k}[n], \forall k,n,$ obtained with \textbf{Algorithm} \ref{sca} is not guaranteed to have a rank equal to or smaller than one, which mandates the use of rank-one approximation or Gaussian randomization procedures \cite{semirank}. In this paper, Gaussian randomization is used to obtain $\mathbf{w}_{k}[n], \forall k,n,$ if the solution $\mathbf{W}_{k}[n], \forall k,n,$ has a rank higher than one. 
\end{remark}
\section{Performance Evaluation}\vspace{-0.1cm}
In this section, we provide simulation results to evaluate the effectiveness of the proposed resource allocation algorithm design for secure MISO-URLLC systems. In our simulations,  a single cell is considered with inner radius $r_{1}=50~\textrm{m}$ and outer radius $r_{2}=500~\textrm{m}$. The BS is located at the center of the cell. The URLLC users and eavesdroppers are randomly distributed within the inner and the outer radius. The path loss is calculated as  $35.3 + 37.6 \log_{10}(d)$\cite{chsecross}, where $d$ is the distance from the BS to a receiver. The system bandwidth is set to $W= 1~\textrm{MHz}$ and the frame duration is $T_{f}=0.25~\textrm{ms}$. The noise power spectral density is $-173~\textrm{dBm/Hz}$. The parameters of \textbf{Algorithm} \ref{sca} are set as $\beta^{(1)}=1000$, $\beta_{\text{max}}=5000$, and $\eta=1.5$. The channel gains follow independent Rayleigh distributions. All simulation results are averaged over $1000$ realizations of the channels gains and path losses.  
\subsection{Performance Bound and Benchmark Schemes}\vspace{-0.1cm}
 We compare the performance of the proposed resource allocation algorithm design with an upper bound, the proposed algorithm without AN, and a baseline scheme:
\begin{itemize}
	\setlength{\itemsep}{1pt}
	\item {\textbf{Secrecy capacity}}: In this scheme, the secrecy capacity formula for infinite blocklength is used for optimization in (\ref{op1}), i.e., the dispersions in constraint $\mathrm{C1}$ are set to zero, but all other constraints are retained. The formulated optimization problem is solved using SCA to find a sub-optimal solution. This scheme provides a lower bound on the total required transmit power of the system.
		\item {\textbf{Proposed scheme without AN}}: In this scheme, the optimization problem is formulated based on the normal approximation. However, AN is not included in the problem formulation. The formulated optimization problem is solved using SCA to find a sub-optimal solution using an algorithm similar to \textbf{Algorithm} 1.
	\item {\textbf{Baseline scheme:}}
	In this scheme, we employ  maximum ratio transmission beamforming (MRT-BF), where $\mathbf{w}_{k}[n]=\sqrt{p_{k}[n]}\frac{\mathbf{h}_{k}}{||\mathbf{h}_{k}||}$. Then, we optimize the power allocated to $p_{k}[n]$. AN is not injected at the BS. The formulated optimization problem is solved using SCA to find a sub-optimal solution using an algorithm similar to \textbf{Algorithm} 1.
\end{itemize}
\subsection{Simulation Results}
Fig.~\ref{fig1} shows the convergence of the proposed algorithm for different numbers of users $K$, different numbers of antennas $N_{T}$, and different numbers of eavesdroppers $J$. We show the total transmit power as a function of the number of iterations for a given channel realization. As can be observed from Fig.~\ref{fig1}, the proposed algorithm converges to a sub-optimal solution after a finite number of iterations. In particular, the proposed algorithm converges after approximately $5$ iterations for all considered parameter values. Moreover, as can be seen, for the considered cases, the speed of convergence of the  proposed algorithm is not sensitive to the numbers of users, antennas, and eavesdroppers.  
\begin{figure}[t]
	\centering
		\resizebox{1\linewidth}{!}{\psfragfig{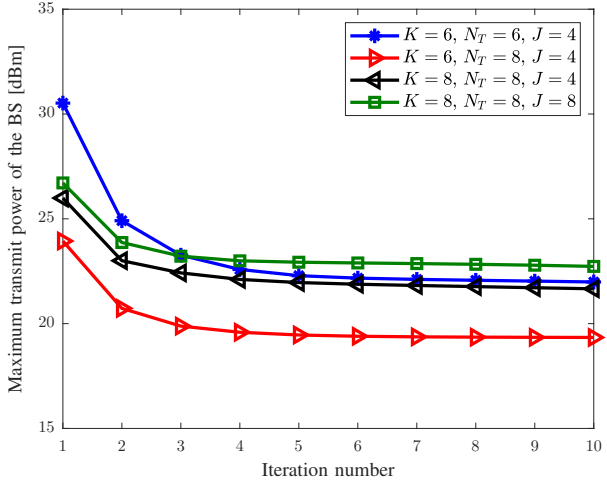}}\vspace{-4mm}
		\caption{Convergence of the proposed algorithm (\textbf{Algorithm} \ref{sca})  for different simulation parameters. $N=4,$ $\epsilon_{k}=10^{-6}, \forall k,$ $\delta=\delta_{j,k}=10^{-6}, \forall j,k,$ $D_{1}=2$, $D_{k}=4,~\forall k \neq 1$.}
		\label{fig1}
\end{figure}

In Fig.~\ref{ant}, we investigate the average transmit power versus the number of antennas at the BS, $N_{T}$, for different resource allocation schemes. As can be observed, the total transmit power at the BS significantly decreases as the number of antennas at the BS increases. This is due to the fact that more antennas offer additional degrees of freedom for resource allocation which facilitate higher received SINRs at the users. The proposed scheme attains large power savings compared to the baseline scheme. The performance loss of the baseline scheme has two reasons. First, the fixed beamformer is strictly sub-optimal. Second, the baseline scheme does not have the capability to impair the eavesdroppers' channel due to the absence of AN. Moreover, the proposed scheme without AN still achieves a good performance compared with the baseline scheme due to the precise beamforming. Furthermore, if the secrecy capacity is used for resource allocation design for URLLC, the required latency,  reliability, and secrecy cannot be guaranteed. This is due to the fact that the performance loss incurred by finite block length coding is not taken into account, and the obtained resource allocation policies do not meet the QoS constraints. 
\begin{figure}[t]
	\centering
	\resizebox{1\linewidth}{!}{\psfragfig{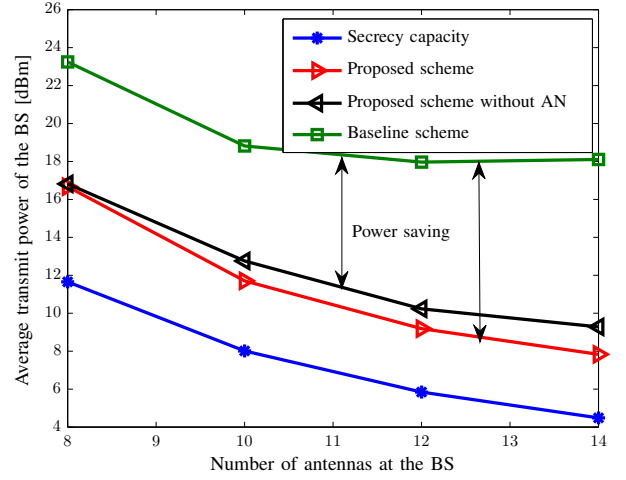}}\vspace{-4mm}
	\caption{Average transmit power versus number of antennas at the BS. $N=4,$ $K=4,$ $J=2,$ $\epsilon_{k}=10^{-7}, \forall k,$ $\delta_{j,k}=10^{-6}, \forall j,k,$ $D_{1}=2,$ $D_{k}=4,~\forall k \neq 1$.} 
	\label{ant}
\end{figure}

In Fig.~\ref{bits}, we show the average transmit power for the proposed scheme versus the required number of secure communication bits, $B^{\text{req}}_{k}, \forall k,$  for different system parameters. In particular, we study the impact of the number of eavesdroppers, $J$, the secrecy constraint on the information leakage, $\delta=\delta_{j,k}, \forall j,k$, and different delay requirements. We consider the following delay scenarios: For delay scenario $S_{1}=\{D_{1}=2, D_{k}=4, \forall k \neq 1\}$, one user has strict delay constraints while the remaining users do not. For delay scenario $S_{2}=\{D_{k}=2, \forall k \in \{1,2,3,4\}, D_{5}=D_{6}=4\}$, four users have strict delay requirements. As expected, increasing the required number of transmitted bits leads to higher transmit powers. This is due to the fact that if more bits are to be transmitted in a frame, higher SINRs are needed for each user, and thus, the BS has to increase the transmitted power. Fig.~\ref{bits} also shows that if the number of eavesdroppers increases, the average transmitted power has to also increase to meet the users' QoS. This stems from the fact that if more eavesdroppers are in the system, the BS has to increase the amount of AN power to degrade the channels of all eavesdroppers to guarantee secrecy. Moreover, as can be observed, the proposed scheme is able to guarantee secure communication even if the number of eavesdroppers exceeds the number of transmit antenna $N_{T}$, due to the precise beamforming and AN design. Furthermore, more stringent constraints on the information leakage $\delta$ increases the required transmit power because more AN power is needed to impair the channel of eavesdroppers to a sufficient degree.

Fig.~\ref{bits} also reveals the effect of the delay constraints. In particular, delay scenario $S_{2}$ leads to a higher power consumption compared to $S_{1}$ because the BS is forced to allocate more power to the  delay sensitive users even if their channel conditions are poor to ensure their transmissions are completed with the desired delay.
\begin{figure}[t]
	\centering
		\resizebox{1\linewidth}{!}{\psfragfig{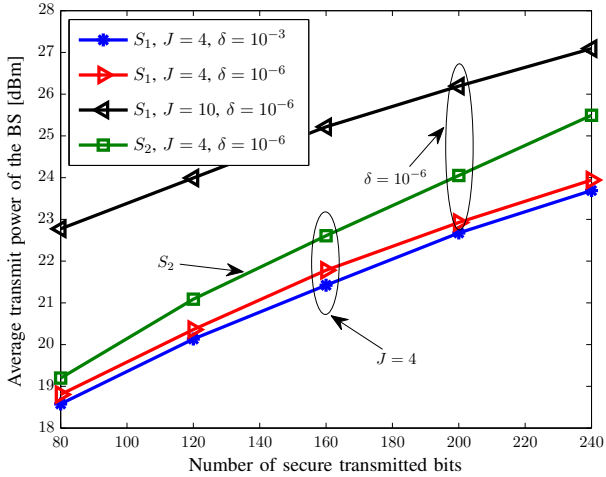}}\vspace{-4mm}
	\caption{Average transmit power versus number of secure bits per packet. $K=6$, $N_{T}=8$, $N=4,$ $\epsilon_{k}=10^{-7}, \forall k$.} 
	\label{bits}
\end{figure}
\section{Conclusion}
In this paper, we studied the resource allocation algorithm design for secure multi-user downlink MISO-URLLC systems. To enhance PLS, AN was injected by the BS to impair the channel of the eavesdroppers. To meet the stringent delay requirements of URLLC users, short packet transmission was adopted and taken into consideration for secure resource allocation design. The resource allocation algorithm design was formulated as an optimization problem for minimization of the total transmit power subject to QoS constraints ensuring the reliability, secrecy, and latency of the URLLC users. The obtained optimization problem was shown to be non-convex but a low-complexity algorithm based on penalized SCA was developed to find a sub-optimal solution. Simulation results showed that the proposed resource allocation algorithm design facilities secure transmission in URLLC systems, and yields a large reduction of the transmit power compared to a baseline scheme. 
\bibliography{ref}  

\begin{thebibliography}{10}
\providecommand{\url}[1]{#1}
\csname url@samestyle\endcsname
\providecommand{\newblock}{\relax}
\providecommand{\bibinfo}[2]{#2}
\providecommand{\BIBentrySTDinterwordspacing}{\spaceskip=0pt\relax}
\providecommand{\BIBentryALTinterwordstretchfactor}{4}
\providecommand{\BIBentryALTinterwordspacing}{\spaceskip=\fontdimen2\font plus
\BIBentryALTinterwordstretchfactor\fontdimen3\font minus
  \fontdimen4\font\relax}
\providecommand{\BIBforeignlanguage}[2]{{%
\expandafter\ifx\csname l@#1\endcsname\relax
\typeout{** WARNING: IEEEtran.bst: No hyphenation pattern has been}%
\typeout{** loaded for the language `#1'. Using the pattern for}%
\typeout{** the default language instead.}%
\else
\language=\csname l@#1\endcsname
\fi
#2}}
\providecommand{\BIBdecl}{\relax}
\BIBdecl

\bibitem{Toward}
G.~Durisi, T.~Koch, and P.~Popovski, ``Toward massive, ultrareliable, and
  low-latency wireless communication with short packets,'' \emph{Proc. {IEEE}},
  vol. 104, no.~9, pp. 1711--1726, Sept 2016.

\bibitem{Popovski1}
P.~Popovski, ``Ultra-reliable communication in {5G} wireless systems,'' in
  \emph{Proc. IEEE Int. Conf. 5G Ubiq. Connect}, Nov 2014, pp. 146--151.

\bibitem{Mehdi1}
M.~Bennis, M.~Debbah, and H.~V. Poor, ``Ultrareliable and low-latency wireless
  communication: Tail, risk, and scale,'' \emph{Proc. IEEE}, vol. 106, no.~10,
  pp. 1834--1853, Oct 2018.

\bibitem{secureofdma}
D.~W.~K. Ng, E.~S. Lo, and R.~Schober, ``Energy-efficient resource allocation
  for secure {OFDMA} systems,'' \emph{{IEEE} Trans. Veh. Commun}, vol.~61,
  no.~6, pp. 2572--2585, July 2012.

\bibitem{misosec1}
Z.~{Chu}, H.~{Xing}, M.~{Johnston}, and S.~{Le Goff}, ``Secrecy rate
  optimizations for a {MISO} secrecy channel with multiple multiantenna
  eavesdroppers,'' \emph{{IEEE} Trans. Wireless Commun}, vol.~15, no.~1, pp.
  283--297, Jan 2016.

\bibitem{misoyan}
Y.~{Sun}, D.~W.~K. {Ng}, J.~{Zhu}, and R.~{Schober}, ``Robust and secure
  resource allocation for full-duplex {MISO} multicarrier {NOMA} systems,''
  \emph{{IEEE} Trans. Commun}, vol.~66, no.~9, pp. 4119--4137, Sep. 2018.

\bibitem{Wyner1}
A.~D. {Wyner}, ``The wire-tap channel,'' \emph{Bell Labs Tech. J.}, vol.~54,
  no.~8, pp. 1355--1387, Oct 1975.

\bibitem{strassen}
V.~Strassen, ``Asymptotische {A}bschatzungen in {S}hannon's
  {I}nformationstheorie,'' \emph{In Proc. 3rd Trans. Prague Conf. Inf. Theory},
  vol.~56, no.~5, pp. 689--723, May 1962.

\bibitem{thesis}
Y.~Polyanskiy, ``Channel coding: {N}on-asymptotic fundamental limits,'' Ph.D.
  dissertation, Princeton University.

\bibitem{Polyanskiy}
Y.~Polyanskiy, H.~V. Poor, and S.~Verdu, ``Channel coding rate in the finite
  blocklength regime,'' \emph{IEEE Trans. Inf. Theory}, vol.~56, no.~5, pp.
  2307--2359, May 2010.

\bibitem{optimal}
Y.~Hu, M.~Ozmen, M.~C. Gursoy, and A.~Schmeink, ``Optimal power allocation for
  {QoS}-constrained downlink multi-user networks in the finite blocklength
  regime,'' \emph{{IEEE} Trans. Wireless Commun}, vol.~17, no.~9, pp.
  5827--5840, Sept 2018.

\bibitem{convexfinite}
S.~Xu, T.~H. Chang, S.~C. Lin, C.~Shen, and G.~Zhu, ``Energy-efficient packet
  scheduling with finite blocklength codes: convexity analysis and efficient
  algorithms,'' \emph{{IEEE} Trans. Wireless Commun}, vol.~15, no.~8, pp.
  5527--5540, Aug 2016.

\bibitem{ghanem1}
W.~Ghanem, V.~Jamali, Y.~Sun, and R.~Schober, ``Resource allocation for
  multi-user downlink {URLLC-OFDMA} systems,'' in \emph{Proc. {IEEE} {Int}.
  {C}ommun. {C}onf.}, Shanghai, P.R. China, May 2019.

\bibitem{ghanem2}
W.~R. Ghanem, V.~Jamali, Y.~Sun, and R.~Schober, ``Resource allocation for
  multi-user downlink {MISO OFDMA-URLLC} systems,'' 2019, Submitted to IEEE
  Trans. Commun., https://arxiv.org/abs/1910.06127.

\bibitem{Quasi}
W.~Yang, G.~Durisi, T.~Koch, and Y.~Polyanskiy, ``Quasi-static multiple-antenna
  fading channels at finite blocklength,'' \emph{IEEE Trans. Inf. Theory},
  vol.~60, no.~7, pp. 4232--4265, July 2014.

\bibitem{wangsecrecy}
W.~{Yang}, R.~F. {Schaefer}, and H.~V. {Poor}, ``Wiretap channels:
  Nonasymptotic fundamental limits,'' \emph{IEEE Trans. Inf. Theory}, vol.~65,
  no.~7, pp. 4069--4093, July 2019.

\bibitem{short1}
H.~{Wang}, Q.~{Yang}, Z.~{Ding}, and H.~V. {Poor}, ``Secure short-packet
  communications for mission-critical {IoT} applications,'' \emph{{IEEE} Trans.
  Wireless Commun}, vol.~18, no.~5, pp. 2565--2578, May 2019.

\bibitem{Csiszar1}
I.~{Csisz\'ar} and J.~{K\"orner}, ``Broadcast channels with confidential
  messages,'' \emph{{IEEE} Trans. Inf. Theory}, vol.~24, no.~3, pp. 339--348,
  May 1978.

\bibitem{Erseghe1}
T.~Erseghe, ``Coding in the finite-blocklength regime: {B}ounds based on
  {L}aplace integrals and their asymptotic approximations,'' \emph{IEEE Trans.
  Inf. Theory}, vol.~62, no.~12, pp. 6854--6883, Dec 2016.

\bibitem{lin1satellite}
Z.~{Lin}, M.~{Lin}, J.~{Ouyang}, W.~{Zhu}, A.~D. {Panagopoulos}, and
  M.~{Alouini}, ``Robust secure beamforming for multibeam satellite
  communication systems,'' \emph{{IEEE} Trans. Veh. Technol.}, vol.~68, no.~6,
  pp. 6202--6206, June 2019.

\bibitem{yan}
Y.~Sun, D.~W.~K. Ng, Z.~Ding, and R.~Schober, ``Optimal joint power and
  subcarrier allocation for full-duplex multicarrier non-orthogonal multiple
  access systems,'' \emph{{IEEE} Trans. Commun}, vol.~65, no.~3, pp.
  1077--1091, March 2017.

\bibitem{Boyed}
S.~Boyd and L.~Vandenberghe, \emph{{Convex Optimization}}.\hskip 1em plus 0.5em
  minus 0.4em\relax New York, NY, USA: Cambridge University Press, 2004.

\bibitem{cvx}
M.~Grant and S.~Boyd, ``{CVX}: Matlab software for disciplined convex
  programming, version 2.1,'' \url{http://cvxr.com/cvx}, Mar. 2014.

\bibitem{localcon}
Q.~T. Dinh and M.~Diehl, ``Local convergence of sequential convex programming
  for nonconvex optimization,'' in \emph{Recent Advances in Optimization and
  its Applications in Engineering}.\hskip 1em plus 0.5em minus 0.4em\relax
  Springer, 2010, pp. 93--102.

\bibitem{pccp}
T.~Lipp and S.~Boyd, ``Variations and extension of the convex--concave
  procedure,'' \emph{Optim. Eng.}, vol.~17, no.~2, pp. 263--287, Jun 2016.

\bibitem{semirank}
Z.~{Luo}, W.~{Ma}, A.~M. {So}, Y.~{Ye}, and S.~{Zhang}, ``Semidefinite
  relaxation of quadratic optimization problems,'' \emph{IEEE Signal Process.
  Mag.}, vol.~27, no.~3, pp. 20--34, May 2010.

\bibitem{chsecross}
C.~She, C.~Yang, and T.~Q.~S. Quek, ``Cross-layer optimization for
  ultra-reliable and low-latency radio access networks,'' \emph{{IEEE} Trans.
  Commun}, vol.~17, no.~1, pp. 127--141, Jan 2018.

\end{thebibliography}
\bibliographystyle{IEEEtran}
\end{document}